\documentclass[aps,prl,reprint,letterpaper,showpacs,twocolumn]{revtex4}
 
\usepackage{amsmath, amsfonts, amssymb, amsthm}
\usepackage{graphicx}

\newtheorem{theorem}{Theorem}
\newtheorem{lemma}{Lemma}

\newtheorem{remark}{\indent \bf Remark}

\newcommand{\Hilbert}{\mathbb}

\newcommand{\DHy}{D_{\textnormal{H}}}

\newcommand*{\trace}{\mathrm{tr}}
\newcommand*{\set}[1]{\mathcal{#1}}
\newcommand*{\ket}{\rangle}
\newcommand*{\bra}{\langle}

\begin{document}

\title{One-Shot Classical-Quantum Capacity and Hypothesis Testing}

\author{Ligong Wang}
\email{wlg@mit.edu}
\affiliation{Research Laboratory of Electronics,
    MIT, Cambridge, MA, USA}

\author{Renato Renner}
\email{renner@phys.ethz.ch}
\affiliation{Institute for Theoretical Physics, ETH Zurich, Switzerland}


\begin{abstract}
  The one-shot classical capacity of a quantum channel quantifies the
  amount of classical information that can be transmitted through a
  single use of the channel such that the error probability is below a
  certain threshold. In this work, we show that this capacity is well
  approximated by a relative-entropy-type measure defined via
  hypothesis testing.  Combined with a quantum version of Stein's
  Lemma, our results give a conceptually simple proof of the
  well-known Holevo-Schumacher-Westmoreland Theorem for the capacity
  of memoryless channels. More generally, we obtain tight capacity
  formulas for arbitrary (not necessarily memoryless) channels.
\end{abstract}


\maketitle

In Information Theory, a \emph{channel} models a physical device that
takes an input and generates an output. One may, for instance, think
of a communication channel (such as an optical fiber) that connects a
sender (who provides the input) with a receiver (who obtains an
output, which may deviate from the input). Another example
is a memory device, such as a hard drive, where the input consists of
the data written into the device, and where the output is the
(generally noisy) data that is retrieved from the device at a later
point in time.

A central question studied in Information Theory is whether, and how,
a channel can be used to transmit data reliably in spite of the
channel noise. This is usually achieved by \emph{coding},
where an \emph{encoder} prepares the channel input by adding
redundancy to the data to be transmitted, and where a \emph{decoder}
reconstructs the data from the noisy channel output.

Here we focus on the case of
\emph{classical-quantum channel coding}, where the data to be
transmitted reliably are classical. No assumptions are made about the
channel that is used to achieve this task, i.e., the inputs and
outputs may be arbitrary quantum states. However, since the
quantum-mechanical structure of the input space is irrelevant for the
encoding of classical data, it can be represented by a (classical) set
$\set{X}$. For any input $x \in \set{X}$, the channel produces an
output, specified by a density operator $\rho_x$ on a Hilbert space
$\Hilbert{B}$. For our purposes, it is therefore sufficient to
characterize a channel by a mapping $x \mapsto \rho_x$ from a set
$\set{X}$ to a set of density operators.

Classical-quantum channel coding has been studied extensively in a
scenario where the channel can be used arbitrarily many times.  The
\emph{channel coding theorem for stationary memoryless
  classical-quantum channels}, established by Holevo~\cite{holevo98}
and Schumacher and Westmoreland~\cite{schumacherwestmoreland97},
provides an explicit formula (see \eqref{eq:HolSchWes}) for the rate
at which data can be transmitted under the assumption that each use of
the channel is independent of the previous uses. More general channel
coding theorems that do not rely on this independence assumption have
been developed in later work by Hayashi and
Nagaoka~\cite{hayashinagaoka03} and by Kretschmann and
Werner~\cite{kretschmannwerner05}. These 
results are asymptotic, i.e., they refer to a limit where the number
of channel uses tends to infinity while the probability of error is
required to approach zero.

Here we consider a scenario where a given quantum channel is used only
\emph{once} and derive tight bounds on the number of classical bits
that can be transmitted with a given average error probability
$\epsilon$, in the following referred to as the
\emph{$\epsilon$-one-shot classical-quantum capacity}. This one-shot
approach provides a high level of generality, as nothing needs to be
assumed about the structure of the channel~\footnote{In particular, in
  contrast to previous work, there is no need no define channels as
  sequences of mappings.}.  (Note that any situation in which a
channel is used repeatedly can be equivalently described as one single
use of a larger channel.) In particular, our bounds on the channel
capacities imply the aforementioned Holevo-Schumacher-Westmoreland
Theorem for the capacity of memoryless channels, as well as the
generalizations by Hayashi and Nagaoka. On the other hand, our work
generalizes similar one-shot results for classical
channels~\cite{rennerwolfwullschleger06,WLGcolbeckrenner09,polyanskiypoorverdu08}. 
Despite their 
generality, the bounds as well as their proofs are
remarkably simple. We hope that our approach may therefore also be of
pedagogical value.

Our derivation is based on the idea, already exploited in previous
works (see, e.g.,
\cite{ogawanagaoka00,ogawanagaoka02,hayashinagaoka03,hayashi07}), of
relating the problem of channel coding to hypothesis testing. Here, we
use hypothesis testing directly to define a
relative-entropy-type quantity, denoted
$\DHy^\epsilon(\cdot\|\cdot)$ (see \eqref{eq:DHdef}). Our main result
asserts that the one-shot channel capacity is well 
approximated by $\DHy^\epsilon(\cdot\|\cdot)$ (Theorem~\ref{thm:main}).

We note that one-shot capacity bounds that are very similar to ours have been implicitly used in the information-spectrum approach to classical-quantum channel coding by Hayashi and Nagaoka \cite{hayashinagaoka03,hayashi06}.

The remainder of this Letter is structured as follows. We briefly
describe hypothesis testing and state a few 
properties of the quantity $\DHy^\epsilon(\cdot \| \cdot)$. We then
state and prove our main result which provides upper and lower
bounds on the $\epsilon$-one-shot classical-quantum capacity in
terms of $\DHy^\epsilon(\cdot \| \cdot)$. 
Finally, we show how the known
asymptotic bounds (for arbitrarily many channel uses) can be obtained
from Theorem~\ref{thm:main}.


\emph{Hypothesis Testing and
  $\DHy^\epsilon(\cdot\|\cdot)$.|}Hypothesis testing is the task of 
distinguishing two possible states of a system, $\rho$ and $\sigma$.
A strategy for this task is specified by a Positive Operator Valued
Measure (POVM) with two elements, $Q$ and $I- Q$, corresponding to the
two possible values for the guess. The probability that the strategy
produces a \emph{correct} guess on input $\rho$ is given by $\trace[ Q
\rho ]$, and the probability that it produces a \emph{wrong} guess on
input $\sigma$ is $\trace[ Q \sigma$].  We define the
\emph{hypothesis testing relative entropy}
$\DHy^\epsilon(\rho\|\sigma)$ as
\begin{align} \label{eq:DHdef}
  \DHy^\epsilon(\rho\|\sigma) \triangleq - \log_2 \inf_{\substack{Q: 0\le
      Q\le I,\\ \trace[Q\rho]\ge 1-\epsilon}} \trace[Q\sigma].
\end{align}
Note that $\DHy^\epsilon(\rho\|\sigma)$ is a semidefinite program and
can therefore be evaluated efficiently. 

As its name suggests, $\DHy^\epsilon(\rho\|\sigma)$ can be understood as a
relative entropy. In particular, for $\epsilon = 0$, it is equal to
\emph{R\'enyi's relative entropy of order $0$}, $D_0(\rho\|\sigma) = -
\log_2 \trace[\rho^0 \sigma]$, where $\rho^0$ denotes the projector
onto the support of $\rho$. For $\epsilon > 0$, it corresponds to a
``smoothed'' variant of the relative R\'enyi entropy of order $0$ used
by Buscemi and Datta~\cite{buscemidatta10} for characterizing the
quantum capacity of channels~\footnote{It is also similar to a
  relative-entropy-type quantity used in~\cite{datta09}, although the
  precise relation to this quantity is not
  known.}. $\DHy^\epsilon(\rho\|\sigma)$ has the following properties,
all of which hold for all $\rho$, $\sigma$ and $\epsilon\in[0,1)$:

\begin{enumerate}
    \item \emph{Positivity:} 
      \begin{align*}
        \DHy^\epsilon(\rho\|\sigma)\ge 0,
      \end{align*}
      with equality if $\rho=\sigma$ and $\epsilon=0$.
    \item \emph{Data Processing Inequality (DPI):} for any Completely
      Positive Map 
      (CPM) $\mathcal{E}$, 
      \begin{align*}
        \DHy^\epsilon(\rho\|\sigma) \ge \DHy^\epsilon
        (\mathcal{E}(\rho)\|\mathcal{E}(\sigma)).
      \end{align*}
    \item Let $D(\cdot\|\cdot)$ denote the usual quantum relative
      entropy, then
      \begin{equation}\label{eq:upper_by_D}
        \DHy^\epsilon(\rho\|\sigma) \le \bigl(D(\rho\|\sigma) +
        H_\textnormal{b}(\epsilon)\bigr)/(1-\epsilon),
      \end{equation}
      where $H_\textnormal{b}(\cdot)$ is the binary entropy function.
 \end{enumerate}

 Positivity follows immediately from the definition. 

 To prove the DPI,
 consider any POVM to distinguish $\mathcal{E}(\rho)$ from
 $\mathcal{E}(\sigma)$. We can then construct a new POVM to
 distinguish $\rho$ from $\sigma$ by preceding the given POVM with the
 CPM $\mathcal{E}$. This new POVM clearly gives the same error
 probabilities (in distinguishing $\rho$ and $\sigma$) as the original
 POVM (in distinguishing $\mathcal{E}(\rho)$ and
 $\mathcal{E}(\sigma)$). The DPI then follows because an optimization
 over all possible strategies for distinguishing $\rho$ and $\sigma$
 can only decrease the failure probability.

 To prove \eqref{eq:upper_by_D}, first see that it holds when
 $D(\rho\|\sigma)$ is replaced by $D(P_\rho\|P_\sigma)$, where
 $P_\rho$ is the distribution of the outcomes of the optimal POVM
 performed on $\rho$, namely, it is $(1-\epsilon, \epsilon)$, and
 similarly for $P_\sigma$ which is $(2^{-\DHy^\epsilon(\rho\|\sigma)},
 1-2^{-\DHy^\epsilon(\rho\|\sigma)})$. This can be shown by directly
 computing $D(P_\rho\|P_\sigma)$. Then \eqref{eq:upper_by_D} follows
 because $D(\cdot\|\cdot)$ satisfies the DPI so $D(\rho\|\sigma) \ge
 D(P_\rho\|P_\sigma)$.

 A further connection between $\DHy^\epsilon(\cdot\|\cdot)$ and
 $D(\cdot\|\cdot)$ is the Quantum Stein's Lemma
 \cite{hiaipetz91,ogawanagaoka00}, which we restate as follows.

 \begin{lemma}[Quantum Stein's Lemma] \label{lem:stein} For any two
   states $\rho$ and $\sigma$ on a Hilbert space and for any
   $\epsilon\in(0,1)$, 
  \begin{align*}
    \lim_{n\to\infty} \frac{1}{n} \DHy^\epsilon (\rho^{\otimes n}\|
    \sigma^{\otimes n}) = D(\rho\|\sigma).
  \end{align*}
\end{lemma}


\emph{Statement and Proof of the Main Result.|}Before stating our main
result, we introduce some general
terminology.  
The encoder is specified by a list of inputs, $\{x_i\}$,
$i\in\{1,\ldots,m\}$, called a \emph{codebook} of size $m$. The
decoder applies a corresponding \emph{decoding POVM}, which acts on
$\Hilbert{B}$ and has $m$ elements. A decoding error occurs if the
output of the decoding POVM is not equal to the index $i$ of the input
$x_i$ fed into the channel. An \emph{$(m,\epsilon)$-code} consists of
a codebook of size $m$ and a corresponding decoding POVM such that,
when the message is chosen uniformly, the average probability of a
decoding error is at most $\epsilon$~\footnote{It is well known that,
  in single-user scenarios, the 
  (asymptotic) capacity does not depend on whether the average (over
  uniformly chosen messages) or the maximum probability of error is
  considered. In the one-shot case, one can construct a code that has
  maximum probability of error not larger than $2\epsilon$ from a code
  with average probability of error $\epsilon$, thereby sacrificing
  one bit.}.

The main result of this Letter is the following theorem.
\begin{theorem} \label{thm:main} The $\epsilon$-one-shot
  classical-quantum capacity of a channel $x \mapsto \rho_x$, i.e.,
  the largest number $R$ for which a $(2^R,\epsilon)$-code exists,
  satisfies 
\begin{multline} \label{eq:mainapprox}
  \sup_{P_X} \DHy^\epsilon(\pi^{\Hilbert{A} \Hilbert{B}} \|
  \pi^{\Hilbert{A}} \otimes \pi^{\Hilbert{B}})  \geq R \\ \geq \sup_{P_X} \DHy^{\epsilon'}(\pi^{\Hilbert{A} \Hilbert{B}} \|
  \pi^{\Hilbert{A}} \otimes \pi^{\Hilbert{B}})  - \log_2
  \frac{4\epsilon}{(\epsilon-\epsilon')^2} \  
\end{multline}
for any $\epsilon'\in (0,\epsilon)$, 
where $\pi^{\Hilbert{AB}}$ is the joint state of the input and output
for an input chosen according to the distribution $P_X$, i.e.,
\begin{align*}
  \pi^{\Hilbert{AB}}  \triangleq  \sum_{x\in\set{X}} P_X(x)
  |x\ket\bra x|^\Hilbert{A}\otimes\rho_x^\Hilbert{B} \ ,
\end{align*}
for any representation of the inputs $x$ in terms of orthonormal
vectors $|x \ket^\Hilbert{A}$ on a Hilbert space $\Hilbert{A}$, and
where $\pi^{\Hilbert{A}}$ and $\pi^{\Hilbert{B}}$ are the
corresponding marginals.
\end{theorem}

\begin{remark}
  In an earlier version of this Letter (which appeared in \emph{Physical Review Letters}), the right-hand side (RHS) of \eqref{eq:mainapprox} is given by its special case where $\epsilon'=\epsilon/2$:
\begin{equation}\label{eq:alternative}
	 \sup_{P_X} \DHy^{\epsilon/2} (\pi^{\Hilbert{A} \Hilbert{B}} \|
  \pi^{\Hilbert{A}} \otimes \pi^{\Hilbert{B}}) - \log_2\frac{1}{\epsilon} - 4.
\end{equation}
For practical scenarios where $\epsilon$ is close to zero, the difference between \eqref{eq:alternative} and the RHS of \eqref{eq:mainapprox} is usually small. However, allowing an arbitrary $\epsilon'$ we can use \eqref{eq:mainapprox} to derive the \emph{$\epsilon$-capacity} of a channel, while fixing $\epsilon'=\epsilon/2$ we cannot.
\end{remark}

The proof of Theorem~\ref{thm:main} is divided into two parts, one for
the first inequality 
(referred to as the \emph{converse}) and the other for the second
inequality (the \emph{achievability}). 
We start with the
converse which asserts that, if a $(2^R,\epsilon)$-code exists,
then
  \begin{align}\label{eq:upper}
    R \le \sup_{P_X}\DHy^\epsilon (\pi^{\Hilbert{AB}} \|
    \pi^{\Hilbert{A}}\otimes\pi^{\Hilbert{B}}). 
  \end{align}

  \begin{proof}[Proof of Theorem~\ref{thm:main}|Converse Part]
  By definition, it is sufficient to prove~\eqref{eq:upper} for a
  uniform distribution on the $x$'s used in the codebook, so we
  can focus on states $\pi^{\Hilbert{AB}}$ of the form
  \begin{align*}
    \pi^{\Hilbert{AB}} = 2^{-R} \sum_{i=1}^{2^R} |x_i\ket\bra x_i|\otimes
    \rho_{x_i} \ . 
  \end{align*}

  Note that the decoding POVM combined with the inverse of the
  encoding map (which is classical) can be viewed as a CPM. This CPM
  maps $\pi^\Hilbert{AB}$ to the (classical) state $P_{MM'}$ denoting
  the joint distribution of the transmitted message $M$ and the
  decoder's guess $M'$. Similarly, it maps $\pi^\Hilbert{A}\otimes \pi
  ^\Hilbert{B}$ to $P_M\otimes P_{M'}$. Hence, it follows from the DPI
  for $\DHy^\epsilon(\rho\|\sigma)$ that
  \begin{align*}
    \DHy^\epsilon (P_{MM'}\|P_M\otimes P_{M'}) \le \DHy^\epsilon
    (\pi^\Hilbert{AB}\|\pi^\Hilbert{A}\otimes\pi^\Hilbert{B}) \ .
  \end{align*}
  It thus remains to prove 
  \begin{align}\label{eq:converse1}
    R\le \DHy^\epsilon (P_{MM'}\|P_M\otimes P_{M'}) \ .
  \end{align}
  For this, we consider a (possibly suboptimal) strategy to
  distinguish between $P_{MM'}$ and $P_M\otimes P_{M'}$. The strategy
  guesses $P_{MM'}$ if $M=M'$, and guesses $P_M\otimes P_{M'}$
  otherwise. Using this distinguishing strategy, the probability of
  guessing $P_M\otimes P_{M'}$ on state $P_{MM'}$ is exactly the
  probability that $M\neq M'$ computed from $P_{MM'}$, namely, the
  average probability of a decoding error, and is thus not larger than
  $\epsilon$ by assumption. Furthermore, the probability of guessing
  $P_{MM'}$ on state $P_{M}\otimes P_{M'}$ is given by
  \begin{align*}
   \sum_{i=1}^{2^R} P_M(i)\cdot
    P_{M'}(i) 
    = 2^{-R} \sum_{i=1}^{2^R} P_{M'}(i) 
    = 2^{-R} \ .
  \end{align*}
  This implies~\eqref{eq:converse1}.
\end{proof}

We proceed with the achievability part of Theorem~\ref{thm:main}.
We show that, for any
$\epsilon>\epsilon'>0$ and $c>0$, there exists a $(2^R,\epsilon)$-code
with
\begin{align}\label{eq:lower}
  R \ge \sup_{P_X} \DHy^{\epsilon'} (\pi^\Hilbert{AB}\|
  \pi^\Hilbert{A}\otimes \pi^{\Hilbert{B}})
  -\log_2\frac{2+c+c^{-1}}{\epsilon-(1+c)\epsilon'} \ .
\end{align}
Optimized over $c$, this bound implies
the second inequality of~\eqref{eq:mainapprox}.

The main technique we need for proving~\eqref{eq:lower} is the
following lemma by Hayashi and Nagaoka \cite[Lemma
2]{hayashinagaoka03}:

\begin{lemma}\label{lem:HayashiNagaoka}
  For any positive real $c$ and any operators $0\le S \le I$ and $T\ge
  0$, we have
  \begin{multline*}
    \lefteqn{I- (S+T)^{-1/2} S (S+T)^{-1/2}}\\
    \le  (1+c)(I-S) + (2+c+c^{-1})T.
  \end{multline*}
\end{lemma}

\begin{proof}[Proof of Theorem~\ref{thm:main}|Achievability Part]
  Fix $\epsilon'>0$, $c>0$, and $P_X$. We are going to
  show that there exists a $(2^R,\epsilon)$-code such that
  \begin{equation*}
    \epsilon \le (1+c)\epsilon' +
     (2+c+c^{-1})2^{R-\DHy^{\epsilon'}(\pi^\Hilbert{AB}\|\pi^\Hilbert{A}\otimes
       \pi^\Hilbert{B})} \ ,
  \end{equation*}
  which immediately implies~\eqref{eq:lower}. 

  Let $Q$ be an operator acting on $\Hilbert{AB}$ such that $0\le Q
  \le I$ and $\trace\big[Q\pi^{\Hilbert{AB}}\big]\ge 1-\epsilon'$. By
  definition, it suffices to prove that there exists a codebook
  and a decoding POVM with error probability
   \begin{align}
     \epsilon \le (1+c)\epsilon' +
     (2+c+c^{-1})2^R\trace\big[Q(\pi^\Hilbert{A}\otimes
       \pi^\Hilbert{B})\big] \ .
     \label{eq:lower18}
   \end{align}

   We generate a codebook by choosing its codewords $x_j$ at random,
   each independently according to the distribution $P_X$.
   Furthermore, we define the corresponding decoding POVM by its
   elements,
  \begin{align*}
    E_i = \left(\sum_{j=1}^{2^R} A_{x_j} \right)^{-\frac{1}{2}}A_{x_i}
    \left(\sum_{j=1}^{2^R} A_{x_j} \right)^{-\frac{1}{2}} \ ,
  \end{align*}
  where $A_x \triangleq \textnormal{tr}_\Hilbert{A}\bigl[
  \left(|x\ket\bra x|^\Hilbert{A} \otimes I^\Hilbert{B} \right)
  Q\bigr]$.

  For a specific codebook $\{x_j\}$ and the transmitted codeword
  $x_i$, the probability of error is given by
  \begin{align*}
    \Pr(\textnormal{error}|x_i, \{x_j\}) = \trace[{(I-E_i)\rho_{x_i}}]
    \ .
  \end{align*}
  We now use Lemma \ref{lem:HayashiNagaoka} with $S = A_{x_i}$ and $T
  = \sum_{j \neq i} A_{x_j}$ to bound this by
  \begin{multline*}
    \Pr(\textnormal{error}|x_i, \{x_j\})  \le 
    (1+c)\bigl(1-\trace[A_{x_i}\rho_{x_i}]\bigr)\\
    + (2+c+c^{-1})\sum_{j\neq i}
    \trace[A_{x_j} \rho_{x_i}] .
  \end{multline*}
  Averaging over all codebooks, but keeping the transmitted codeword
  $x_i$ fixed, we find
   \begin{multline*}
     \lefteqn{\Pr(\textnormal{error}|x_i) \le (1+c)
     \bigl(1-\trace[A_{x_i}\rho_{x_i}]\bigr)}\\
     +   (2+c+c^{-1})(2^R-1)\trace \Big[\sum_{x'\in\set{X}} P_X(x')
         A_{x'}\rho_{x_i} \Big].
   \end{multline*}
   Averaging now in addition over the transmitted codeword $x_i$, we
   obtain the upper bound
   \begin{multline}
     \Pr(\textnormal{error})  \le (1+c)\bigl(1-\sum_x P_X(x)
       \trace[A_x\rho_x]\bigr) \\
      + (2+c+c^{-1})2^R 
     \trace\Big[\sum_{x'}
         P_X(x') A_{x'} \sum_x P_X(x) \rho_x \Big] .
     \label{eq:lower_1}
   \end{multline}
   Note that
   \begin{multline*}
      \sum_x P_X(x) \trace[A_x \rho_x]
   = 
      \sum_x P_X(x) \trace\Big[Q |x\ket\bra x|^\Hilbert{A} \otimes
        \rho_x^\Hilbert{B} \Big] \\
   =
     \trace\big[Q\pi^\Hilbert{AB}\big]
   \geq 
     1-\epsilon'
   \end{multline*}
   and 
   \begin{multline*}
     \trace{\Big[\sum_{x} P_X(x') A_{x'}\sum_x P_X(x) \rho_x\Big]} \\
    = 
    \sum_{x'}P_X(x') \trace \Big[Q |x'\ket\bra x'| \otimes 
       \sum_x P_X(x) \rho_x\Big]\\
    =
     \trace\big[Q(\pi^\Hilbert{A}\otimes
         \pi^\Hilbert{B})\big] .
   \end{multline*}
   Inserting these expressions into~\eqref{eq:lower_1} we find that
   the upper bound \eqref{eq:lower18} holds for the probability of
   error averaged over the class of codebooks we generated.  Thus
   there must exist at least one codebook whose error probability
   $\epsilon$ satisfies \eqref{eq:lower18}. 
\end{proof}


\emph{Asymptotic Analysis.|}Theorem~\ref{thm:main} applies to the
transmission of a message in a single use of the channel. Obviously, a
channel that can be used $n$ times can always be modeled as one big
single-use channel. We can thus retrieve the known expressions for the
(usual) capacity of channels, i.e., the average number of bits that
can be transmitted per channel use in the limit where the channel is
used arbitrarily often and the error $\epsilon$ approaches $0$. Most
generally, a channel that can be used an arbitrary number of times is
characterized by a sequence of mappings 
$x_n\mapsto \rho^n$, $n\in\{1,2,\ldots\}$, where $x_n \in \set{X}_n$
represents an input 
state over $n$ channel-uses \footnote{Here $\set{X}_n$ cannot be
  replaced by $\set{X}^{\times n}$, which only represents
  \emph{product} states.}, and where $\rho^n$ is a density operator on
$\Hilbert{B}^{\otimes n}$. Note that such a channel need not have any
structure such as ``causality'' as defined in
\cite{kretschmannwerner05}. From 
Theorem \ref{thm:main} it immediately follows that the capacity of any
channel is given by
\begin{align} \label{eq:general}
  C = \lim_{\epsilon\downarrow 0} \varliminf_{n\to \infty} \frac{1}{n} \sup_{P_{X_n}}
  \DHy^\epsilon(\pi^{\Hilbert{A}_n\otimes \Hilbert{B}^{\otimes n}}\|
  \pi^{\Hilbert{A}_n} \otimes \pi^{\Hilbert{B}^{\otimes n}}),
\end{align}
where $\Hilbert{A}_n$ denotes the Hilbert space spanned by orthonormal
states $|x_n\ket$ for all $x_n\in\set{X}_n$. This expression is
equivalent to 
\cite[Theorem~1]{hayashinagaoka03}~\footnote{While the expressions
  used in~\cite{hayashinagaoka03} look completely different, one can
  (non-operationally) prove their equivalence. See
  \cite{WLG11}.}. We
can also derive similar results for the \emph{optimistic capacity} and
the \emph{$\epsilon$-capacity}, see~\cite{WLG11}.

Now consider a \emph{memoryless} channel whose
behavior in each use is independent of the previous uses. 
The capacity $C$ of such a channel is given by the well-known
Holevo-Schumacher-Westmoreland Theorem
\cite{holevo98,schumacherwestmoreland97}:
\begin{align} \label{eq:HolSchWes}
  C =  \lim_{k \to \infty} \frac{1}{k} \sup_{P_{X_k}}
  D(\pi^{\Hilbert{A}_k\otimes\Hilbert{B}^{\otimes
      k}}\|\pi^{\Hilbert{A}_k}\otimes 
  \pi^{\Hilbert{B}^{\otimes k}})  \ .
\end{align}
Note that $D(\pi^{\Hilbert{A}_k\otimes\Hilbert{B}^{\otimes
    k}}\|\pi^{\Hilbert{A}_k}\otimes
\pi^{\Hilbert{B}^{\otimes k}}) $ may equivalently be written as a
mutual information $I(\Hilbert{A}_k;\Hilbert{B}^{\otimes
  k})$. This theorem can be proved easily using
\eqref{eq:general}. 

\begin{proof}[Proof of \eqref{eq:HolSchWes}]
  To show achievability, i.e., that $C$ is lower-bounded by the
  RHS of \eqref{eq:HolSchWes}, we restrict the
  supremum in \eqref{eq:general} to product distributions on $k$-use
  states, so the joint state $\pi^{\Hilbert{A}_n\otimes
    \Hilbert{B}^{\otimes n}}$ looks like $(\pi^{\Hilbert{A}_k\otimes
    \Hilbert{B}^{\otimes k}})^{\otimes ( n/k )}$~\footnote{As $n$
    tends to infinity, the problem that $n$ might not be divisible
    by $k$ becomes negligible.}. We then let $n$ tend to infinity and
  apply Lemma~\ref{lem:stein} to obtain that, for any $k$,
  \begin{align} \label{eq:k}
    C \ge \frac{1}{k}
    \sup_{P_{X_k}} D(\pi^{\Hilbert{A}_k\otimes
        \Hilbert{B}^{\otimes k}} \|
      \pi^{\Hilbert{A}_k}\otimes \pi^{\Hilbert{B}^{\otimes
          k}}). 
  \end{align}
  This concludes the proof of the achievability part. 

  The converse, i.e., that $C$ is upper-bounded by the RHS of
  \eqref{eq:HolSchWes}, follows immediately from \eqref{eq:general} and
  \eqref{eq:upper_by_D}. 
\end{proof}

To conclude, it may be interesting to compare Theorem~\ref{thm:main}
to other recently derived bounds on the one-shot capacity of
classical-quantum channels~\cite{mosonyidatta09,mosonyihiai11,renesrenner11}. The
bounds of~\cite{mosonyidatta09,mosonyihiai11} are different from ours in that they
are not known to coincide asymptotically for arbitrary channels.
In~\cite{renesrenner11}, it has been shown that the one-shot
classical-quantum capacity $R$ of a channel can be approximated (up to
additive terms of the order $\log_2 1/\epsilon$) by
\begin{align*} 
   R \approx \max_{P_X}  H_{\min}^\epsilon(A)_{\pi^{\Hilbert{A} }}  - H_{\max}^\epsilon(A|B)_{\pi^{\Hilbert{A} \Hilbert{B}}} \ ,
\end{align*}
where $H_{\min}^\epsilon$ and $H_{\max}^\epsilon$ denote the smooth
min- and max-entropies, which have recently been shown to be the
relevant quantities for characterizing a number of
information-theoretic tasks (see, e.g.,
\cite{tomamichelcolbeckrenner09} for definitions and
properties). Combined with our result, this suggests that there is a
deeper and more general relation between hypothesis testing and smooth
entropies (and, therefore, the associated operational
quantities). Exploring this link is left as an open question
for future work.

\emph{Acknowledgments.|}The authors thank Masahito Hayashi and Marco Tomamichel for their comments on an earlier version of this Letter. L.W. acknowledges support from the US Air
Force Office of Scientific Research (grant No.  FA9550-11-1-0183) and
the National Science Foundation (grant No. CCF-1017772). R.R.
acknowledges support from the Swiss National Science Foundation (grant
Nos.\ 200021-119868, 200020-135048, and the NCCR ``QSIT'') and the
European Research 
Council (ERC) (grant No.\ 258932).


%

\end{document}